\documentclass[conference
]{IEEEtran}

\usepackage[noadjust]{cite}
\usepackage[dvips]{graphicx}

\usepackage{mathrsfs}
\usepackage{amssymb}
\usepackage{amsfonts}
\usepackage{bbold}

\usepackage{mathrsfs}
\usepackage{amsthm}
\usepackage{amsmath}
\newtheorem{statement}{Proposition}
\newtheorem{theorem}{Theorem}
\newtheorem{example}{Example}

\def\rank{\mathop{\rm rank}}
\def\mod{\mathop{\rm mod}}
\def\wgt{\mathop{\rm wgt}}
\def\Hn{\mathcal{H}_2^{\otimes n}}

\def\mat#1{\mathcal{#1}}

\newcommand{\ket}[1]{\left\vert{#1}\right\rangle}

\makeatother

\begin{document}
%
\title{Improved quantum hypergraph-product LDPC codes}

\date\today
\author{\IEEEauthorblockN{Alexey A.\ Kovalev\IEEEauthorrefmark{1} and
    Leonid P.\ Pryadko\IEEEauthorrefmark{2}}
  \IEEEauthorblockA{Department of Physics \& Astronomy, University of
    California, Riveside, California 92521}
  \IEEEauthorblockA{\IEEEauthorrefmark{1}%
    Email:{alexey.kovalev@ucr.edu}
    \qquad \IEEEauthorrefmark{2} 
Email:{leonid.pryadko@ucr.edu} }}

\maketitle
\begin{abstract}
  We suggest several techniques to improve the toric codes and the finite-rate
  generalized toric codes (quantum hypergraph-product codes) recently
  introduced by Tillich and Z\'emor.  For the usual toric codes, we introduce
  the rotated lattices specified by two integer-valued periodicity vectors.
  These codes include the checkerboard codes, and the family of minimal
  single-qubit-encoding toric codes with block length $n=t^2+(t+1)^2$ and
  distance $d=2t+1$, $t=1,2,\ldots$.  We also suggest several related
  algebraic constructions which increase the rate of the existing
  hypergraph-product codes by up to four times.
\end{abstract}

\maketitle
\IEEEpeerreviewmaketitle

\section{Introduction}

Quantum error
correction\cite{shor-error-correct,Knill-Laflamme-1997,Bennett-1996}
made quantum computing (QC) theoretically possible.  However, high
precision required for error correction
\cite{Knill-error-bound,
  Dennis-Kitaev-Landahl-Preskill-2002,%
  Steane-2003,%
  Fowler-QEC-2004,Fowler-2005,
  Raussendorf-Harrington-2007} combined
with the large number of auxiliary qubits necessary to implement it,
have so far inhibited any practical realization beyond
proof-of-the-principle
demonstrations\cite{chuang-2000,chuang-2001,Gulde-2003,%
  Chiaverini-2004,Friedenauer-2008,
  Kim-2010}.  

For stabilizer codes, the error syndrome is obtained by measuring the
generators of the stabilizer group.  The corresponding quantum measurements
can be greatly simplified (and also done in parallel) in low-density
parity-check (LDPC) codes which are specially designed to have stabilizer
generators of small weight.  Among LDPC codes, the toric (and related surface)
codes \cite{kitaev-anyons,Dennis-Kitaev-Landahl-Preskill-2002,%
  Raussendorf-Harrington-2007,Bombin-2007} have the stabilizer generators of
smallest weight, $w=4$, with the support on neighboring sites of a
two-dimensional lattice.  These codes have other nice properties which make
them suitable for quantum computations with relatively high error threshold.
Unfortunately, these code families have very low code rates that scale as
inverse square of the code distance.

Recently, Tillich and Z\'emor proposed a finite-rate generalization of toric codes\cite{Tillich2009}.  The construction relates a quantum code
to a direct product of hypergraphs corresponding to two classical
binary codes.  Generally, thus obtained LDPC codes have finite rates and the
distances that scale as a square root of the block length.
Unfortunately, despite finite asymptotic rates, for smaller
block length, the rates of the quantum codes which can be obtained from the
construction\cite{Tillich2009} are small.

In this work, we present a construction aimed to improve the rates of
both regular toric\cite{kitaev-anyons} and generalized toric
codes\cite{Tillich2009}.  For the toric codes, we introduce the
rotated tori specified by two integer-valued periodicity vectors.
Such codes include the checkerboard codes \cite{Bombin-2007}
($\pi/4$-rotation), and the family \cite{Kovalev-Dumer-Pryadko-2011}
of minimal single-qubit-encoding toric codes with block length
$n=t^2+(t+1)^2$ and distance $d=2t+1$, $t=1,2,\ldots$.  For the
generalized toric codes\cite{Tillich2009}, we suggest an algebraic
construction equivalent to the $\pi/4$ rotation of the regular toric
codes.  The resulting factor of up to four improvement of the code
rate makes such codes competitive even at relatively small block
sizes.

\section{Definitions.}
We consider binary quantum error correcting codes (QECCs) defined on the
complex Hilbert space $\Hn$ where $\mathcal{H}_{2}$ is the complex Hilbert
space of a single qubit $\alpha\left|0\right\rangle +\beta\left|1\right\rangle
$ with $\alpha,\beta\in\mathbb{C}$ and
$\left|\alpha\right|^{2}+\left|\beta\right|^{2}=1$. Any operator acting on
such an $n$-qubit state can be represented as a combination of Pauli
operators which form the Pauli group $\mathscr{P}_{n}$ of size $2^{2n+2}$ with
the phase multiplier $i^{m}$:
\begin{equation}
  \mathscr{P}_{n}=i^{m}\{I,X,Y,Z\}^{\otimes n},\; m=0,\ldots,3\:,
  \label{eq:PauliGroup}
\end{equation}
where $X$, $Y$, and $Z$ are the usual Pauli matrices and $I$ is the
identity matrix.  It is customary to map the Pauli operators, up to a
phase, to two binary strings, $\mathbf{v},\mathbf{u}\in\{0,1\}^{\otimes n}$
\cite{Calderbank-1997},
\begin{equation}
  U\equiv i^{m'}X^{\mathbf{v}}Z^{\mathbf{u}}\: \rightarrow(\mathbf{v},\mathbf{u}),
  \label{eq:mapping}
\end{equation}
where $X^{\mathbf{v}}=X_{1}^{v_{1}}X_{2}^{v_{2}}\ldots X_{n}^{v_{n}}$ and
$Z^{\mathbf{u}}=Z_{1}^{u_{1}}Z_{2}^{u_{2}}\ldots Z_{n}^{u_{n}}$.  A product of two
quantum operators corresponds to a sum ($\mod 2$) of the corresponding 
pairs $(\mathbf{v}_i,\mathbf{u}_i)$.

An $[[n,k,d]]$ stabilizer code $\mathcal{Q}$ is a $2^k$-dimensional
subspace of the Hilbert space $\Hn$ stabilized by an Abelian
stabilizer group $\mathscr{S}=\left\langle G_{1},\ldots
  ,G_{n-k}\right\rangle $, $-{\mathbb 1}\not\in\mathscr{S}$
\cite{gottesman-thesis}.  Explicitly,
\begin{equation}
  \label{eq:stabilizer-code}
  \mathcal{Q}=\{\ket\psi: S\ket\psi=\ket\psi,\forall S\in \mathscr{S}\}.
\end{equation}
Each generator $G_i\in \mathscr{S}$ is mapped according to
Eq.~(\ref{eq:mapping}) in order to obtain the binary check matrix
$H=(A_{X}|A_{Z})$ in which each row corresponds to a generator, with
rows of $A_{X}$ formed by $\mathbf{v}$ and rows of $A_{Z}$ formed by
$\mathbf{u}$ vectors.  For generality, we assume that the matrix
$H$ may also contain unimportant linearly dependent rows which are
added after the mapping has been done.  The commutativity of
stabilizer generators corresponds to the following condition on the
binary matrices $A_{X}$ and $A_{Z}$:
\begin{equation}
A_{X}A_{Z}^{T}+A_{Z}A_{X}^{T}=0 \;(\mod 2).\label{eq:product}
\end{equation}

A more narrow set of Calderbank-Shor-Steane (CSS) codes
\cite{Calderbank-Shor-1996} contains codes whose stabilizer generators can be
chosen to contain products of only Pauli $X$ or Pauli $Z$
operators.  For these codes the parity check matrix can be chosen in the form:
\begin{equation}
H=\left(\begin{array}{c|c}
G_{X} & 0\\
0 & G_{Z}
\end{array}\right),\label{eq:CSS}
\end{equation}
where the commutativity condition simplifies to $G_{X}G_{Z}^{T}=0$.  

The dimension of a quantum code is $k=n-\rank H$; for a CSS code this
simplifies to $k=n-\rank G_{X}-\rank G_{Z}$.  

The distance $d$ of the quantum code is given by the minimum weight of an
operator $U$ which commutes with all operators from the stabilizer
$\mathscr{S}$, but is not a part of the stabilizer, $U\not\in \mathscr{S}$.
In terms of the binary vector pairs $(\mathbf{a},\mathbf{b})$, this is
equivalent to a minimum weight of the bitwise OR $(\mathbf{a}|\mathbf{b})$ of
all pairs satisfying the symplectic orthogonality condition,
\begin{equation}
  A_X \mathbf{b}+A_Z \mathbf{a}=0,\label{eq:symplectic}  
\end{equation}
which are not linear combinations of the rows of $H$.

\section{Toric codes and rotated toric codes}
\subsection{Canonical construction}

We consider the toric codes\cite{kitaev-anyons} in the restricted sense, with qubits located
on the bonds of a square lattice $L_\xi\times L_\eta$, with periodic boundary
conditions along the directions $\xi$ and $\eta$.  The stabilizer generators
$A_i\equiv \prod_{j\in\square_i}X_j $ and $B_i\equiv \prod_{j\in +_i}Z_j$ are
formed as the products of $X_j$ around each plaquette, and $Z_j$ around each
vertex (this defines a CSS code).  The corresponding block length is
$n=2L_\xi L_\eta$, and there are $r_A=r_B=L_\xi L_\eta-1$ independent
generators of each kind, which leaves us with the code of size
$k=n-r_A-r_B=2$.  This code is degenerate:
the degeneracy group is formed by products of the generators $A_i$, $B_i$; its
elements can be visualized as (topologically trivial) loops drawn on the
original lattice (in the case of products of $A_i$), or the dual lattice in
the case of products of $B_i$.  The two sets of logical operators are formed
as the products of $X$ ($Z$) operators along the topologically non-trivial
lines formed by the bonds of the original (dual) lattice (see
Fig.~\ref{fig:toric}).  The code distance $d=\min(L_\xi,L_\eta)$ is given by
the minimal weight of such operators.
\begin{figure}[htbp]
  \centering
  \includegraphics[scale=1.41]{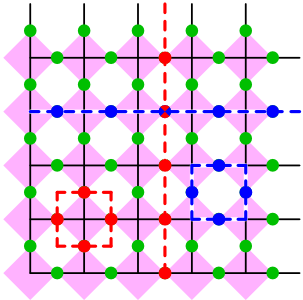}
  \caption{(Color online) Lattice representing the canonical toric code
    $[[50,2,5]]$.  The generators $A_i$ are formed by Pauli $X$ generators
    around a plaquette (blue square) while the generators $B_i$ are formed by
    Pauli $Z$ generators around a vertex (red square).  Dashed horizontal blue
    line and vertical red line represent a pair of mutually conjugate logical
    operators formed by the products of $X$ and $Z$ respectively.  Shading
    corresponds to an alternative checkerboard representation of the
    underlying lattice.}
  \label{fig:toric}
\end{figure}
\subsection{Checkerboard codes \protect\cite{Bombin-2007}}
\label{sec:checkerboard}
In the following, it will be convenient to consider a lattice with qubits
placed on the vertices.  Then, if we color every other plaquette to form a
checkerboard pattern, we can define the operators $A_i$ as products of $X$
operators around the colored plaquettes, and the operators $B_i$ as products
of $Z$ operators around the white plaquettes (see Fig.~\ref{fig:checker},
Left).  Now, the checkerboard code with $n=L_x L_y$, where both $L_x$ and
$L_y$ are even, can be defined by taking periodic boundary conditions on the
sides of a rectangle of size $L_x\times L_y$.  The condition ensures that we
can maintain a consistent checkerboard pattern.  Then, the product of all
$A_i$ (or of all $B_i$) gives identity.  Thus, the stabilizer is formed by
$n-2$ independent generators, which again gives $k=2$ as in the regular toric
codes.  The two sets of logical operators are formed by the products of $X$
operators along the topologically non-trivial paths drawn through the colored
areas, and the products of $Z$ operators along the topologically non-trivial
paths through the white areas (see Fig.~\ref{fig:checker}, Left).  The
distance of the code, $d=\min(L_x,L_y)$, corresponds to the shortest
topologically non-trivial chain of qubits, graphically, a horizontal or
a vertical straight line.

\begin{figure}[htbp]
  \centering
  \includegraphics[scale=1.0]{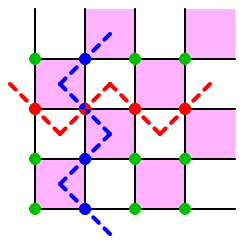}\hskip0.5in
  \includegraphics[scale=1.0]{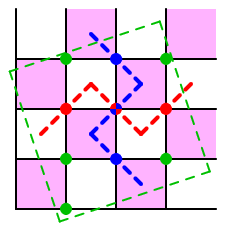}
  \caption{Left: Lattice representation of the checkerboard code $[[16,2,4]]$.
    Qubits are placed in the lattice vertices; dashed blue and red lines
    represent a pair of logical operators as in Fig.~\ref{fig:toric}.  Right:
    same for the rotated checkerboard code $[[10,2,3]]$.}
  \label{fig:checker}
\end{figure}

\subsection{Checkerboard codes with arbitrary rotation}
Compared to the regular toric codes, the checkerboard codes use half as many
qubits with the same $k$ and distance.  The disadvantage is that the distance
is always even.  This latter restriction can be lifted by introducing
periodicity vectors which are not necessarily parallel to the bonds of the
lattice.  (Note that a similar trick was used in early small-cluster exact
diagonalization studies of the Hubbard model\cite{Dagotto-Joynt-etal-1990}).

Let us define two integer-valued periodicity vectors $\mathbf{L}_i=(a_i,b_i)$, $i=1,2$,
and identify all points on the lattice which can
be connected by a vector of the form $m_1 \mathbf{L}_1+m_2 \mathbf{L}_2$, with
integer $m_i$.  The checkerboard pattern is preserved iff both
$\|\mathbf{L}_i\|_1\equiv |a_i|+|b_i|$ are even, $i=1,2$.  Such a cluster
contains
\begin{equation}
  n=|\mathbf{L}_1\times \mathbf{L}_2|=|a_1 b_2 -b_1 a_2 |
\label{eq:cluster-size}
\end{equation}
vertices, and, again, we have $k=2$ as for the standard checkerboard codes.

Since the qubits in the positions shifted by $\mathbf{L}_i$ are the same, it
is easy to see that our code is identical to that on a cluster with
periodicity vectors, e.g., $\mathbf{L}_1$, $\mathbf{L}_1+\mathbf{L}_2$, and,
generally, a cluster with periodicity vectors
$\mathbf{L}_i'=g_{ij}\mathbf{L}_j$, where the integer-valued matrix $g_{ij}$
has the determinant $\det g=\pm 1$.

For a periodicity vector $\mathbf{L}=(a,b)$ with $a+b$ even, the shortest
topologically non-trivial qubit chain has $\|\mathbf{L}\|_\infty\equiv
\max(|a|,|b|)$ operators which leads to the code distance:
\begin{equation}
d(\mathbf{L}_1,\mathbf{L}_2)=\min_{m_1,m_2}
\|m_1\mathbf{L}_1+m_2\mathbf{L}_2\|_\infty. \label{eq:dist-checker}
\end{equation}


\begin{example}
  A family of near-optimal odd-distance checkerboard codes can be introduced
  by taking $\mathbf{L}_1=(2t+1,1)$, $\mathbf{L}_2=(-1,2t+1)$, $t=1,2,\ldots$.
  Such codes have the parameters $[[1+(2t+1)^2, 2,2t+1]]$;
  explicitly: $[[10,2,3]]$ (illustrated in Fig.~\ref{fig:checker}, Right),
  $[[26,2,5]]$, $[[50,2,7]]$, \ldots.
\end{example}

\begin{example}
  \label{ex:toric}
  The original toric codes are recovered by taking $\mathbf{L}_i$ along the
  diagonals, $\mathbf{L}_1=(L_\xi,L_\xi)$, $\mathbf{L}_2=(-L_\eta,L_\eta)$, so
  that $\|L_i\|_1$ are always even, thus $n=2L_\xi L_\eta$, $k=2$, and
  $d=\min(L_\xi,L_\eta)$.  For odd distances, taking $L_\xi=L_\eta=d$, we have
  the codes $[[18,2,3]]$, $[[50,2,5]]$, $[[98,2,7]]$, \ldots.
\end{example}

\subsection{Non-bipartite rotated toric codes}
We now construct a version of rotated toric codes on clusters with at least
one of the periodicity vectors $\mathbf{L}_i$ violating the checkerboard
pattern, e.g., $\|\mathbf{L}_1\|_1$ odd.  Since the checkerboard pattern cannot be
maintained, we define identical stabilizer generators in a non-CSS form, with the stabilizer
generators $G_i=ZXXZ$ on each plaquette given by the products of $Z$ operators
along one diagonal, and $X$ operators along the other diagonal.  With periodic
boundary conditions, the product of all $G_i$ is an identity, and this is the
only relation between these operators on a non-bipartite cluster.  Thus, here
we have only one encoded qubit, $k=1$.

The operators $G_i$ can be viewed as local Clifford (LC) transformed
$A_i$ or $B_i$ operators of the toric code.  It is easy to see that
the logical operators have to correspond to topologically non-trivial
closed chains of qubits, as for the bipartite case.  However, in order
to close the loop, we have to take only the translation vectors with
$\|\mathbf{L}\|_1$ even.  For example, if $\|\mathbf{L}_1\|_1$ is odd
and particularly small, the minimal chain could wrap twice around the
direction given by $\mathbf{L}_1$.  Since the two turns could share
some of the qubits, it is difficult to come up with a general
expression for the distance.

\begin{example}
  \label{ex:non-bipartite}
  Checkerboard-like codes can be obtained by taking $L_x$ or $L_y$ odd.
  Smallest codes in this family correspond to $L_x=L_y=d$; they have
  parameters $[[d^2,1,d]]$, where $d=2t+1$.  Explicitly, $[[9,1,3]]$,
  $[[25,1,5]]$, $[[49,1,7]]$, \ldots.
\end{example}
\begin{example}
 \label{ex:toric-rotated}
  A family of smallest odd-distance rotated toric codes
  \cite{Kovalev-Dumer-Pryadko-2011} is obtained for $\mathbf{L}_1=(t+1,t)$,
  $\mathbf{L}_2=(-t,t+1)$, $t=1,2,\ldots$.  These codes have the parameters
  $[[t^2+(t+1)^2,1,2t+1]]$.  Explicitly, $[[5,1,3]]$, $[[13,1,5]]$,
  $[[25,1,7]]$, $[[41,1,9]]$, \ldots.
\end{example}

\section{Generalized toric and checkerboard codes}
\subsection{Algebraic representation of hypergraph-product codes}
\label{sec:algebraic}
The finite-rate generalization\cite{Tillich2009} of the toric code relies
on hypergraph theory, with the square lattice generalized to a product of
hypergraphs (each corresponding to a parity check matrix of a classical binary
code).  We first recast the original construction into an algebraic language.

Let $\mat{H}_1$ (dimensions $r_{1}\times n_{1}$) and $\mat{H}_2$ (dimensions
$r_{2}\times n_{2}$) be two binary matrices.  The associated (hypergraph-product)
quantum code\cite{Tillich2009} is a CSS code with the stabilizer
generators
\begin{equation}
\begin{array}{c}
  \displaystyle 
  G_{X}=(E_{2}\otimes\mathcal{H}_{1},\mathcal{H}_{2}\otimes E_{1}),\\
  \displaystyle 
  G_{Z}=(\mathcal{H}_{2}^{T}\otimes\widetilde{E}_{1},\widetilde{E}_{2}
  \otimes\mathcal{H}_{1}^{T}). 
\end{array}\label{eq:Till}
\end{equation}
Here each matrix is composed of two blocks constructed as Kronecker products
(denoted with ``$\otimes$''), and $E_i$ and $\widetilde{E}_i$, $i=1,2$, are
unit matrices of dimensions given by $r_i$ and $n_i$, respectively.  The
matrices $G_X$ and $G_Z$, respectively, have $r_1r_2$ and $n_1n_2$ rows (not
all of the rows are linearly independent), and they both have
$n\equiv r_2 n_1+r_1 n_2$ columns, which gives the block length of the quantum
code.  The commutativity condition $G_X G_Z^T=0$ is obviously satisfied by
Eq.~(\ref{eq:Till}) since the Kronecker product obeys $(A\otimes B)(C\otimes
D)=AC \otimes BD$.

Note that the construction~(\ref{eq:Till}) is somewhat similar to product
codes introduced by Grassl and R\"otteler\cite{Grassl-Rotteler-2005}.  The
main difference is that here the check matrix and not the generator matrix is
written in terms of direct products.

The parameters $[[n,k,d]]$ of thus constructed quantum code are determined by
those of the four classical codes which use the matrices $\mat{H}_1$,
$\mat{H}_2$, $\mat{H}_1^T$, and $\mat{H}_2^T$ as the parity-check matrices.
The corresponding parameters are introduced as
\begin{equation}
  \label{eq:params}
  \mat{C}_{\mat{H}_i}=[n_i,k_i,d_i],\quad 
   \mat{C}_{\mat{H}_i^T}=[{\widetilde
     n}_i,\widetilde{k}_i,\widetilde{d}_i],\quad i=1,2,
\end{equation}
where we use the convention \cite{Tillich2009} that the distance
$d_i(\widetilde{d}_i)= \infty $ if $k_i(\widetilde{k}_i)=0$.
The matrices $\mat{H}_i$ are arbitrary, and are allowed to have
linearly-dependent rows and/or columns.  As a result, both
$k_i=n_i-\rank\mat{H}_i$ and $\widetilde k_i=\widetilde
n_i-\rank\mat{H}_i$ can be non-zero at the same time as the block
length of the ``transposed'' code $\mat{C}_{\mat{H}_i^T}$ is given by
the number of rows of $\mat{H}_i$, $\widetilde n_i=r_i$.

Specifically, for the hypergraph-product code (\ref{eq:Till}), we have
$n=r_2n_1+r_1n_2$, $k=2k_1k_2-k_1s_2-k_2s_1$ with $s_i=n_i-r_i,\, i=1,2 \,$ (Theorem 7 from
Ref.~\cite{Tillich2009}), while the distance $d$ satisfies the conditions
$d\ge \min(d_1,d_2,\widetilde d_1, \widetilde d_2)$ (Theorem 9 from
Ref.~\cite{Tillich2009}), and two upper bounds (Lemma 10 from
Ref.~\cite{Tillich2009}): if $k_1>0$ and $\widetilde k_2>0$, then $d\le d_1$; if
$k_2>0$ and $\widetilde k_1>0$, then $d\le d_2$.

These parameters can also be readily established from the stabilizer
generators in the form of Eq.~(\ref{eq:Till}).  
For example, the dimension of the quantum code follows from 
\begin{statement}\label{prop:full} The number of linearly independent
  rows in matrices $G_X$ and $G_Z$ given by Eq.~ (\ref{eq:Till}) is
  $\rank G_X=r_{1}r_{2}-\widetilde{k}_{1}\widetilde{k}_{2}$ and $\rank
  G_Z =n_{1}n_{2}-k_{1}k_{2}$.
\end{statement}
\begin{proof} The matrices $G_X$ and $G_Z$ have
  $r_1r_2$ and $n_1n_2$ rows, respectively. To count the number of
  linearly-dependent rows in $G_X$, we notice that the equations
  $(a^T\otimes b^T)\cdot (E_2\otimes \mat H_1)=0$ and $(a^T\otimes b^T)\cdot
  (\mat H_2\otimes E_1)=0$ are both satisfied iff $a\in \mat
  C_{\mat H_2^T}$ and $b\in \mat C_{\mat H_1^T}$, thus there are $\widetilde
  k_1 \widetilde k_2$ linear relations between the rows of $G_X$,
  and we are left with $r_1r_2-\widetilde k_1 \widetilde k_2$
  linearly-independent rows.  Similarly, there are
  $n_1n_2-k_1k_2$ linearly independent rows in $G_Z$.
\end{proof}

To prove the lower bound on the distance, consider a vector
$\mathbf{u}$ such that $G_X \cdot \mathbf{u}=0$ and
$\wgt(\mathbf{u})<d$.  We construct a quantum code in the form
(\ref{eq:Till}) from the matrices $\mat{H}_1'$, $\mat{H}_2'$ formed
only by the columns of respective $\mat{H}_i$, $i=1, 2$, that are
involved in the product $G_X \cdot \mathbf{u}$.  According to
Proposition~\ref{prop:full}, the reduced code has $k=0$, so that the
reduced $\mathbf{u}'$, $G_X'\cdot \mathbf{u}'=0$, has to be a linear
combination of the rows of $G_Z'$.  The rows of $G_Z'$ are a subset of
those of $G_Z$, with some all-zero columns removed; thus the full vector $\mathbf{u}$ is also a linear combination of
the rows of $G_Z$.  Similarly, a vector $\mathbf{v}$ such that $G_Z
\cdot \mathbf{v}=0$ and $\wgt(\mathbf{v})<d$, is a linear combination
of rows of $G_X$.

The upper bound is established by considering vectors
$\mathbf{u}\equiv (\mathbf{e}\otimes \mathbf{c},0)$ with
$\mathbf{c}\in\mathcal{C}_{\mat{H}_1}$, which requires $k_{1}>0$.
Vector $\mathbf{e}$, $\wgt(\mathbf{e})=1$, for which $\mathbf{u}$ is
not a linear combination or rows of $G_{Z}$, exists only when
$\widetilde{k}_{2}>0$.  The other upper bound is established by
considering vectors $(0,\mathbf{c}\otimes \mathbf{e})$ with
$\mathbf{c}\in\mathcal{C}_{\mat{H}_2}$.

\subsection{Original code family from full-rank matrices}
\label{sec:orig}
In Ref.~\cite{Tillich2009}, only one large family of
quantum codes based on the hypergraph-product ansatz~(\ref{eq:Till})
is given.  Namely, the matrix $\mathcal{H}_{1}$ is taken as a
full-rank parity matrix of a binary LDPC code with parameters
$\mathcal{C}_{\mat{H}_1}=[n_1,k_1,d_1]$ ($r_1=n_1-k_1$), so that the
transposed code has dimension zero, $\widetilde k_1=0$.  The second
matrix is taken as $\mathcal{H}_{2}=\mathcal{H}_{1}^{T}$, so that
$\mathcal{C}_{\mat{H}_2^T}=\mathcal{C}_{\mat{H}_1}$.  Then
Eq.~(\ref{eq:Till}) defines a quantum LDPC code with parameters
\begin{equation}
  \mathcal{Q}^\mathrm{orig}=[[(n_1-k_1)^{2}+n_1^{2},k_1^{2},d_1]],
  \label{eq:par-orig} 
\end{equation}
where the weight of each row of $G_X$, $G_Z$ equals to the sum of the
row-weight and the column-weight of $\mathcal{H}_1$.

\begin{example}
  Let $\mathcal{H}_1$ be a parity-check matrix of the repetition code
  $[d,1,d]$.  Then the quantum code has the parameters
  $[2d^2-2d+1,1,d]$.  Explicitly, $[[13,1,3]]$, $[[25,1,4]]$,
  $[[41,1,5]]$,\ldots --- these parameters are inferior compared to
  the original toric code family, cf. Examples \ref{ex:non-bipartite},
  \ref{ex:toric-rotated}.
\end{example}

\subsection{Code family from square matrices}
\label{sec:squar}
Instead of using full-rank parity-check matrices\cite{Tillich2009},
let us start with a pair of binary codes with square parity-check
matrices $\mat{H}_i$, such that $\widetilde d_1=d_1$, $\widetilde
d_2=d_2$.  Then, automatically, $\widetilde k_i=k_i=n_i-\rank
\mat{H}_i$.  The hypergraph-product ansatz~(\ref{eq:Till}) gives the
code with the parameters
\begin{equation}
  \mathcal{Q}^{\rm square}=[[2n_1
n_2, 2k_1k_2, \min(d_1,d_2)]].\label{eq:par-squar}
\end{equation}
Note that the rate $R=k/n$ of this family is
up to twice that of the family originally suggested in
Ref.~\cite{Tillich2009}, see Sec.~\ref{sec:orig}.  

\begin{example}
  The standard toric codes are recovered by taking for $\mat{H}_2=\mat{H}_1$
  the circulant matrix of a repetition code. The code parameters are $[[2d_1^2,2,d_1]]$,
  cf.\ Example \ref{ex:toric}.
\end{example}

We suggest two general ways to obtain suitable square parity check matrices.
First, if we start from an $[n_1,k_1,d_1]$ LDPC code with the full-rank parity check matrix $P$, we can construct the following symmetric matrix,
\begin{equation}
  \label{eq:symmetrization}
  \mat{H}_1^{\rm sym}=\left(
  \begin{array}[c]{cc}
    \mathbb1&P\\P^T&0
  \end{array}\right),
\end{equation}
so that the code $\mathcal{C}_{\mat{H}_1^{\rm sym}}$ is a $[2n_1-k_1,k_1,d_1]$ LDPC code.  

Second construction assumes that $\mathcal{C}_{\mat{H}_i}$ are cyclic LDPC codes.  The full
circulant matrices $\mat{H}_i$ are constructed from coefficients of check polynomials $h_i(x)$.  The check polynomials of the
transposed code, $\widetilde h_i(x)=h_i(x^{n_i-1})\mod (x^{n_i}-1)$, are just
the original check polynomials reversed, and the original and transposed codes have the same
parameters.  

\subsection{Code family from symmetric matrices.}
If we have two symmetric parity-check matrices,
$\mat{H}_i=\mat{H}_i^T$, $i=1,2$ [e.g., from
Eq.~(\ref{eq:symmetrization})], the full hypergraph-product
code~(\ref{eq:Till}) can be transformed into a direct sum of two
independent codes, each with the following non-CSS check matrix
\begin{equation}
  \label{eq:check-symm}
  H=(E_2\otimes\mat{H}_1|\mat{H}_2\otimes E_1), \quad
  \mat{H}_i^T=\mat{H}_i,\quad i=1,2.
\end{equation}This gives the following
 \begin{theorem}\label{th:symmetrized}A quantum code
in Eq.~(\ref{eq:check-symm}) has parameters
\begin{equation}
  \mathcal{Q}^{\rm sym}=[[n_1
n_2, k_1k_2, \min(d_1,d_2)]].\label{eq:par-symm}
\end{equation}
\end{theorem}  Thus, we can reduce by half
both the blocklength and the number of encoded qubits, i.e., keeping the
rate of Eq.~(\ref{eq:par-squar}) but doubling the relative distance.

For a cyclic LDPC code $\mathcal{C}_{\mat{H}}$ with a \emph{palindromic}
check polynomial, $x^{\deg h(x)}h(1/x)=h(x)$, such that $n-\deg h(x)$
is even, we can always construct a symmetric circulant matrix
$\mat{H}$ from the polynomial $x^{[n-\deg h(x)]/2}h(x)$.
\begin{example} If $\mathcal{H}_{1}=\mathcal{H}_{2}$ are symmetric
  check matrices of a cyclic $[n_1,k_1,d_1]$ code corresponding to a
  palindromic polynomial $h(x)$, then the quantum code has parameters
  $[[n_1^{2},k_1^{2},d_1]]$.  In particular, for $n_1=17$ and
  $h(x)=1+x^3+x^4+x^5+x^6+x^9$ we obtain $[[289,81,5,w=12]]$ code, and
  for $h(x)=1+x$, we recover the non-bipartite checkerboard codes from
  Example \ref{ex:non-bipartite}.
\end{example}

\subsection{Code family from two-tile codes}

Finally, let us construct a generalization of the regular {}``bipartite''
checkerboard codes. We start with a pair of binary codes with the
parity check matrices of even size 
\begin{equation}
  \mathcal{H}_{1}=\tbinom{1 0}{0 1} \otimes a_{1}+\tbinom{0 1}{1 0}\otimes
  b_{1}, 
  \quad\mathcal{H}_{2}^{p}=a_{2}\otimes\tbinom{1 0}{0 1}+b_{2}\otimes\tbinom{0 1}{1 0},
  \label{eq:check-tiled}  
\end{equation}
constructed from the half-size matrices  ({}``tiles'') $a_{i}$, $b_{i}$ with the
distances of the classical codes
$\mathcal{C}_{\mathcal{H}_{i}}$ and
$\mathcal{C}_{\mathcal{H}_{i}^{T}}$ given by $d_{i}$ and
$\widetilde{d}_{i}$, $i=1,2$, where the check matrix
$\mathcal{H}_{2}=\tbinom{1 0}{0 1} \otimes a_{2}+\tbinom{0 1}{1 0}\otimes
  b_{2}$ is
equivalent to $\mathcal{H}_{2}^{p}$ and can be rendered to the latter
form by row and column permutations. It is convenient to introduce
notation for the dimensionality of symmetric subspaces of
$\mathcal{C}_{\mathcal{H}_{1}}$ and $\mathcal{C}_{\mathcal{H}_{2}^p}$
containing only words of type ${1 \choose 1}\otimes\alpha_{1}$ and
$\alpha_{2}\otimes{1 \choose 1}$ as $k_{i}^{s}\equiv n_{i}/2-\rank
(a_{i}+b_{i})$, and for asymmetric subspaces as
$k_{i}^{a}\equiv k_{i}-k_{i}^{s}$, $i=1,2$ (analogously we define
$\widetilde{k}_{i}^{s}$ and $\widetilde{k}_{i}^{a}$).  We define
half-size CSS matrices {[}cf.~Eq.~(\ref{eq:Till}){]}
\begin{equation}
\begin{array}{c}
  G_{X}=(E_{2}^{({1/2})}\otimes\mathcal{H}_{1},\mathcal{H}_{2}^{p}
  \otimes E_{1}^{({1/2})}),\\ 
  G_{Z}=(\mathcal{H}_{2}^{pT}\otimes\widetilde{E}_{1}^{({1/2})},
  \widetilde{E}_{2}^{({1/2})}\otimes\mathcal{H}_{1}^{T}),
\end{array}\label{eq:Tor1}
\end{equation}
where the identity matrices $E_{i}^{(1/2)}$,
$\widetilde{E}_{i}^{(1/2)}$ have dimensions $r_{i}/2$, $n_{i}/2$,
half-size compared to those in Eq.~(\ref{eq:Till}).
\begin{statement}\label{th:lemma}
  The numbers of linearly independent rows in matrices~(\ref{eq:Tor1})
  are $\rank
  G_X=r_{1}r_{2}/2-\widetilde{k}_{1}^{s}\widetilde{k}_{2}^{s}-
  \widetilde{k}_{1}^{a}\widetilde{k}_{2}^{a}$ and $\rank G_Z
  =n_{1}n_{2}/2-k_{1}^{s}k_{2}^{s}-k_{1}^{a}k_{2}^{a}$.
\end{statement} \begin{proof} To count the number of
  linearly-dependent rows in $G_X$, we notice that the equations
  $\upsilon^{T}\cdot(E_{2}^{({1/2})}\otimes\mathcal{H}_{1})=0$ and
  $\upsilon^{T}\cdot(\mathcal{H}_{2}\otimes E_{1}^{({1/2})})=0$ are
  both satisfied for ansatz \begin{equation}
    \upsilon=\alpha_{1}\otimes{\alpha_3 \choose
      \alpha_4}+\alpha_{2}\otimes{\alpha_4 \choose \alpha_3},
\end{equation}
if and only if either (i) $\alpha_{1}\neq\alpha_{2}$,
$\alpha_{3}\neq\alpha_{4}$ and ${\alpha_1 \choose
  \alpha_2}\in\mathcal{C}_{\mathcal{H}_{2}^{T}}$, ${\alpha_3 \choose
  \alpha_4}\in\mathcal{C}_{\mathcal{H}_{1}^{T}}$ or (ii)
$\upsilon=\alpha_{1}'\otimes{1 \choose 1}\otimes\alpha_{3}'$ and
$\alpha_{1}'\in\mathcal{C}_{a_{2}^{T}+b_{2}^{T}}$,
$\alpha_{3}'\in\mathcal{C}_{a_{1}^{T}+b_{1}^{T}}$, thus there are
$\widetilde{k}_{1}^{s}
\widetilde{k}_{2}^{s}+\widetilde{k}_{1}^{a}\widetilde{k}_{2}^{a}$
linear relations between the rows in $G_X$, and we are left with
$\rank G_{X}=r_{1}r_{2}/2-\widetilde{k}_{1}^{s}
\widetilde{k}_{2}^{s}-\widetilde{k}_{1}^{a}\widetilde{k}_{2}^{a}$
linearly-independent rows. Similarly, we prove that
$\rank G_{Z}=n_{1}n_{2}/2-k_{1}^{s}k_{2}^{s}-k_{1}^{a}k_{2}^{a}$.
\end{proof} 
\begin{theorem}\label{th:tiled}
  A quantum CSS code in Eqs.~(\ref{eq:check-tiled}) and
  (\ref{eq:Tor1}) has the parameters:
  \begin{equation}
    \begin{array}{c}
      n=(n_{1}r_{2}+n_{2}r_{1})/2,\\
      k=2 k_{1}^s k_{2}^s+2 k_{1}^{a}k_{2}^{a}-k_{1}s_{2}/2-k_{2}s_{1}/2,\\
      d\geq\min(d_{1}/2,d_{2}/2,\widetilde{d}_{1}/2,\widetilde{d}_{2}/2),
    \end{array}\label{eq:parameters}
  \end{equation}
  where $s_{i}=n_{i}-r_{i}$, $i=1,2$. In addition, for $k_{1}>0$ and
  $\widetilde{k}_{2}>0$ the upper bound $d\le d_{1}$ exists and for
  $k_{2}>0$ and $\widetilde{k}_{1}>0$ the upper bound $d\le d_{2}$
  exists.
\end{theorem} 
\begin{proof} The number of encoded qubits $k$ follows from
  Proposition \ref{th:lemma}.  The lower bound on the distance can be
  established as for the original hypergraph-product codes in
  Sec.~\ref{sec:algebraic}, except now the reduced binary check
  matrices $\mat{H}_1'$, $\mat{H}_2'$ should preserve the tiled form
  (\ref{eq:check-tiled}).  Hence, for every column involved in
  the product $G_X\cdot \mathbf{u}$, we may need to insert two columns
  into the reduced matrices; thus we need $\wgt(\mathbf{u})<d/2$ which
  reduces the lower bound on the distance.  The two upper bounds can
  be established by considering vectors $(\mathbf{e}\otimes
  \mathbf{c},0)$ with $\mathbf{c}\in\mathcal{H}_{1}$ and
  $(0,\mathbf{c}\otimes \mathbf{e})$ with
  $\mathbf{c}\in\mathcal{H}_{2}$, exactly as for the
  hypergraph-product codes in Sec.~\ref{sec:algebraic}.
\end{proof}
\begin{theorem}
\label{th:tiled1}
Suppose $a_{i}$ and $b_{i}$, $i=1,2$ in Eq.~(\ref{eq:check-tiled}) are
such that $k_{i}^{a}=0$, $k_{i}^{s}\neq 0$, $r_{i}=n_{i}$ and binary
codes with generator matrices $a_{i}+b_{i}$ and $a_{i}^T+b_{i}^T$ are
not distance $1$ codes. Then the quantum code in Eq.~(\ref{eq:Tor1})
has parameters
$[[n_{1}n_{2},2k_{1}k_{2},\min(d_{1},d_{2},\widetilde{d}_{1},\widetilde{d}_{2})]]$,
cf.\ Eq.~(\ref{eq:par-squar}).
\end{theorem}
The proof is similar to the proof of Theorem \ref{th:tiled}. The
additional restrictions on the binary codes guarantee that a vector
$\mathbf{u}$ of weight less than $d$ can only overlap with columns of
$\mat{H}_i$ in less than $d$ positions even after the symmetric
counterparts are added.

If we start from distance-$d$ LDPC codes with half size square parity matrices $\mat{H}_i^{(1/2)}$ [e.g., from Eq. (\ref{eq:symmetrization})] then $a_i=\mat{H}_i^{(1/2)}+E^{(1/2)}$ and $b_i=E^{(1/2)}$ in Eq. (\ref{eq:check-tiled}) lead to distance-$2d$ code satisfying Theorem \ref{th:tiled1}.
Alternatively, one can start with two cyclic LDPC codes with even blocksize $n_{i}$,
$i=1,2$, and the check polynomials $h_{i}(x)$ that divide
$x^{n_{i}/2}-1$. The corresponding square circulant parity-check
matrices $\mathcal{H}_{1}$ and $\mathcal{H}_{2}$ (and
$\mathcal{H}_{2}^p$) satisfy~(\ref{eq:check-tiled}). The generator
polynomials,
\begin{equation}
g_{i}(x)=(x^{n_{i}}-1)/h_{i}(x)=(x^{n_{i}/2}+1)\,(x^{n_{i}/2}-1)/h_{i}(x),\label{eq:generator-tiled}
\end{equation}
 and their reversed indicate that $k_i^a=0$. 
\begin{example}If $\mathcal{H}_{1}$
is the square parity matrix of a cyclic $[n_1,k_1,d_1]$ code corresponding to the polynomial $h(x)$ that divides $1-x^{n_1/2}$
and $\mathcal{H}_{2}=\mathcal{H}_{1}$ then the quantum code has parameters $[[n_1^{2},2k_1^{2},d_1]]$.  For $n_1=30$ and $h(x)=1+x+x^3+x^5$ we obtain $[[900,50,14,w=8]]$ code. For $h(x)=1+x$, we recover
the bipartite checkerboard codes from Sec.~\ref{sec:checkerboard}.\end{example}

\section{Conclusions}
We suggested several simple techniques to improve existing quantum LDPC codes,
toric codes, and generalized toric codes with asymptotically finite rate
(quantum hypergraph-product codes\cite{Tillich2009}).  In the latter case we
increased the rate of the code family originally proposed in
Ref.~\cite{Tillich2009} by up to four times.

\section*{Acknowledgment}
We are grateful to I. Dumer and M. Grassl for multiple helpful discussions.
This work was supported in
part by the U.S. Army Research Office under Grant No.\ W911NF-11-1-0027, and by the
NSF under Grant No. 1018935.





%


\end{document}